\documentclass[12pt,a4paper,reqno]{amsart}
\usepackage{amsmath,amsthm,amssymb}
\usepackage{graphicx} 
\usepackage[noadjust]{cite}
\usepackage{amsthm}
\usepackage{bm}
\usepackage[top=30truemm,bottom=30truemm,left=23truemm,right=23truemm]{geometry} 
\usepackage{url}
\usepackage{multirow}

\numberwithin{equation}{section}

\theoremstyle{plain}
\newtheorem{thm}{Theorem}[section]
\newtheorem{lem}[thm]{Lemma}
\newtheorem{prop}[thm]{Proposition}
\newtheorem{cor}[thm]{Corollary}

\theoremstyle{definition}

\theoremstyle{remark}

\newcommand{\tr}{\text{tr}}
\newcommand{\hf}{\frac{1}{2}}

\begin{document}

\title[Spatially homogeneous solutions of vacuum Einstein equations]{Spatially homogeneous solutions of vacuum Einstein equations in general dimensions}
\author{Yuichiro Sato \and Takanao Tsuyuki}
\date{\today}

\address{Global Education Center, Waseda University, 
Nishi-waseda, 1-6-1, Shinjuku-ku, Tokyo, 169-0805, Japan}

\email{yuichiro-sato@aoni.waseda.jp}

\address{Faculty of Business Administration and Information Science, Hokkaido Information University, 
Nishi-nopporo 59-2, Ebetsu, Hokkaido, 069-8585, Japan}

\email{tsuyuki@do-johodai.ac.jp}



\begin{abstract}
We study time-dependent compactification of extra dimensions.
We assume that the spacetime is spatially homogeneous, and solve the vacuum Einstein equations without cosmological constant in more than three dimensions. 
We consider globally hyperbolic spacetimes in which almost abelian Lie groups act on the spaces isometrically and simply transitively. 
We give left-invariant metrics on the spaces and solve Ricci-flat conditions of the spacetimes. 
In the four-dimensional case, our solutions correspond to the Bianchi type~II solution.
By our results and previous studies, 
all spatially homogeneous solutions whose spaces have zero-dimensional moduli spaces of left-invariant metrics are found.
For the simplest solution, we show that each of the spatial dimensions cannot expand or contract simultaneously in the late-time limit.
\end{abstract}

\maketitle

\section{Introduction}
Why the number of the observed spatial dimensions is three?
In standard physics, such as the standard model of particle physics and general relativity, there is no reason that the number has to be three.
Indeed, the extra dimensions have been considered in various theories such as Kaluza-Klein \cite{Nordstrom:1914ejq, Kaluza:1921tu,Klein:1926tv,Witten:1981me}, supergravity \cite{Cremmer:1978km}, superstring \cite{Green:2012oqa}, and M-theories \cite{Witten:1995ex} (for recent studies on Einstein equations in 10-dimensional supergravity, see \cite{Tsuyuki:2021xqu,Takeuchi:2023egl}).
The superstring theories and the M-theory are particularly interesting, since they can predict the number of spatial dimensions as nine and ten, respectively.
If the extra dimensions exist but are undetectably small, why they are much smaller than the others? 

We expect that it is the result of the time evolution of the universe.
The evolution can be described by the Einstein equations.
They are difficult to be solved in general, so some assumptions are necessary.
In the previous studies of time-depending extra dimensions, spacetimes have been often assumed to be the products of two spaces of constant curvature \cite{MR898660,Townsend:2003fx,Ohta:2003pu,Russo:2019fnk}. 
There are only two scale factors in such cases.
The advantage of this assumption is that the Einstein equations become simple.
The disadvantage is that the setup seems special and artificial.
We do not make such an assumption in this paper; the scale factors of each direction can evolve differently.
Instead, we assume that the spacetime is spatially homogeneous.
This assumption reduces Einstein equations into ordinary differential equations (ODEs). 
The observed universe is homogeneous at a sufficiently large scale, so we expect that the assumption is natural as a first approximation.

Spatially homogeneous but anisotropic spacetimes in four dimensions are known as the Bianchi spacetimes. 
A Bianchi spacetime is a four-dimensional Lorentzian 
manifold that has a spacelike hypersurface on which a Lie group simply transitively acts. 
In particular, it is expressed as a product $I\times G$ of an open interval $I$, which describes the time, and a connected Lie group $G$, which describes the space. 
Three-dimensional Lie groups are classified into types I to IX.
Types I to VII are solvable (see Table.~\ref{tBianchi}), while types VIII and IX are semisimple. 
In addition, in the three-dimensional case, a Lie group is almost abelian if and only if it is solvable \cite{MR2530885}.

In this study, we consider vacuum Einstein equations without the cosmological constant.
In the four-dimensional case, general vacuum solutions have been found only for the types I \cite{MR1506447}, II \cite{Taub:1950ez} and V \cite{MR191573} (see also \cite{MR2003646}).
These types correspond to $G=\mathbb{R}^3$ (abelian group), $H_3$ (Heisenberg group) and $\mathbb{R}\mathrm{H}^3$ (hyperbolic space), respectively. 
The higher dimensional generalization of type I ($G=\mathbb{R}^{n}$) and V ($G=\mathbb{R}\mathrm{H}^{n}$) have been done in \cite{Chodos:1979vk} and \cite{Demaret:1985pt}.
Ref.~\cite{Demaret:1985pt} also studies other higher-dimensional vacuum solutions including Bianchi type III and VI, but metrics are assumed to be diagonal.
Higher dimensional generalizations of the Bianchi type IX were also studied in \cite{Ishihara:1985eu}.
The main purpose of this paper is to derive the vacuum solution for the higher dimensional extension of Bianchi type II.
However, the way of the extension is not unique, so we consider the moduli space $\mathfrak{PM}$ of left-invariant metrics.

It is known that $\dim \mathfrak{PM}$ is zero if and only if the Lie group is $\mathbb{R}^3$, $H_3$, or $\mathbb{R}\mathrm{H}^3$ \cite{MR2783396,MR1958155}. 
Accidentally or not, these groups correspond to those of the Bianchi types that vacuum solutions are known.
In the higher $n$-dimensional case ($n\ge 4$), 
we already have the classification of simply-connected Lie groups 
whose moduli spaces are zero-dimensional \cite{MR1958155}.
These are given by $\mathbb{R}^{n}$, $H_{3} \times \mathbb{R}^{n-3}$, $\mathbb{R}\mathrm{H}^{n}$
and these are all almost abelian (see Theorem~\ref{classification}).
For this reason, we study the case with $G=H_{3} \times \mathbb{R}^{n-3}$ 

We solve the vacuum Einstein equations for $G=H_{3} \times \mathbb{R}^{n-3}$. As a result, the Ricci-flat metric is:
\begin{equation}
\begin{split}
ds^2 =& \cosh(kt)\left(-c_0^2 e^{2(\tr h)t}dt^2+c_2^2e^{2h_2 t}dx_2^2
+c_n^2e^{2h_n t}dx_n^2\right)\\
&+\frac{c_1^2e^{2h_1 t}}{\cosh (kt)}\left(dx_1-x_n dx_2\right)^2
+\sum_{i=3}^{n-1}c_i^2e^{2h_i t}dx_i^2,\\
k=& \frac{c_0 c_1}{c_2 c_n},\ 2h_1 =-\sum_{i=3}^{n-1}h_i,\ (\tr h)^2-(\tr h^2) = \frac{k^2}{2},
\end{split}
\label{eq:ds2}
\end{equation}
where $c_i$ are positive constants and $h_i$ are real constants.
In the $n=3$ case, $\sum_{i=3}^2$ terms just mean zero.
We can let $c_i=1$ by the isometry and rescaling, but we leave it to show that
this solution includes those previously found in $(3+1)$ dimension \cite{Taub:1950ez} and in $(4+1)$ dimension \cite{Christodoulakis:2003wy}. The $(4+1)$-dimensional case was also studied in \cite{Halpern:2002vd}. 
The scale factors $a_i$ of our solution (\ref{eq:ds2}) include their result, but the lapse function $N$ is slightly different.
We find that the method of deriving the above metric can be applied to different Lie algebra cases, and the other new solutions are given in (\ref{eq:ds2c}) and (\ref{eq:ds2e}).


We also prove that the metrics can be diagonalized for the cases $G=\mathbb{R}^{n}$, $H_{3} \times \mathbb{R}^{n-3}$, and $\mathbb{R}\mathrm{H}^{n}$.
Thus, the theorem below is proved.
\begin{thm}
    Except for the Minkowski spacetimes, 
    all general solutions of vacuum Einstein equations for the spatially homogeneous spacetimes $I\times G$, where $G$ has zero-dimensional moduli space of left-invariant Riemannian metrics, are \eqref{eq:ds2}, \eqref{eq:kas}, \eqref{eq:solV} and \eqref{eq:mil}.
\end{thm}

As for the solution \eqref{eq:ds2}, we prove that all of the dimensions cannot expand or contract simultaneously. 
Thus, some spatial dimensions can be much smaller than the other dimensions, even though we did not assume anisotropic energy-momentum tensors.
We also see that the solutions that only three dimensions expand exist. 
However, the number of the expanding dimensions is not limited to three.
The question mentioned at the beginning of this paper is still unanswered.

This paper is organized as follows. 
In Section~\ref{s:pre}, we introduce higher dimensional generalizations of Bianchi spacetimes and give their vacuum Einstein equations explicitly.
In Section~\ref{s:dim}, $\dim \mathfrak{PM} = 0$ solutions are enumerated. 
In particular, the new solution (\ref{eq:ds2}) is derived. 
In Section~\ref{sdiag}, we discuss the diagonalizability of metrics, 
and we show that metrics can be diagonalized for these solutions. 
In Section~\ref{s:exp}, we discuss the behavior of the solution (\ref{eq:ds2}). 
In Appendix~\ref{s:oth}, we show other new solutions of the vacuum Einstein equations that are also generalizations of Bianchi type II.

\begin{table}[tb]
\centering 
\caption{Three-dimensional almost abelian Lie groups classified by Bianchi types.
The rest of the Bianchi types, VIII and IX, are not almost abelian.
The associated matrices $A$ express the non-zero structure constants as \eqref{Lie_alg_str}.
The `Solutions' row shows references for spatially-homogeneous general solutions of vacuum ($n+1$)-dimensional spacetimes.
Higher-dimensional generalization (defined by $G=H_{3} \times \mathbb{R}^{n-3}$) of the vacuum solution for Type II is done in this paper.
Bianchi type VI and VII contain a parameter $c\in \mathbb{R}$.
}
\begin{tabular}{ccccccccc}
\hline\hline
&&I&II&III&IV&V&VI&VII\\  \hline 
\multicolumn{2}{c}{\rule[-5mm]{0mm}{12mm} $A$ in \eqref{Lie_alg_str}} & 
$
\begin{bmatrix}
    0&0\\0&0
\end{bmatrix}$
&
$
\begin{bmatrix}
    0&1\\0&0
\end{bmatrix}$
&$
\begin{bmatrix}
    1&1\\1&1
\end{bmatrix}$
&$
\begin{bmatrix}
    1&1\\0&1
\end{bmatrix}$
&$
\begin{bmatrix}
    1&0\\0&1
\end{bmatrix}$
&$
\begin{bmatrix}
    c&1\\1&c
\end{bmatrix}$
&$
\begin{bmatrix}
    c&-1\\1&c
\end{bmatrix}$ \\ 
\hline
\multirow{2}{*}{Solutions } &$n=3$  & \cite{MR1506447} &
\cite{Taub:1950ez} & & & \cite{MR191573} \\
&$n\geq 4$
 & \cite{Chodos:1979vk} &  \eqref{eq:ds2} & && \cite{Demaret:1985pt} \\
 \hline\hline
\end{tabular}
\label{tBianchi}
\end{table}

\section{Preliminaries}
\label{s:pre}
\subsection{Spatially homogeneous spacetimes}
Let $(M,g_{M})$ be an $(n+1)$-dimensional Lorentzian manifold, where $g_{M}$~(or $ds^{2}$) denotes a Lorentzian metric on $M$.
We call $(M,g_{M})$ \textit{time-orientable} if there exists a global vector field $V$ on $M$ such that
\begin{equation*}
    g_{M}(V,V)<0.
\end{equation*}
In this paper, when $(M,g_{M})$ is connected and time-orientable, we call it an $(n+1)$-dimensional \textit{spacetime}.

Let $(M,g_{M})$ be a spacetime, $\Lambda$ a constant and $T$ a $(0,2)$-type divergence-free symmetric tensor field on $M$.
Then the following equation
\begin{equation}
\textrm{Ric} - \frac{\textrm{Scal}}{2} g_{M} + \Lambda g_{M} = T \label{Ein_eq}
\end{equation}
is called the \textit{Einstein equation}.
Moreover, $\Lambda$ and $T$ are called the \textit{cosmological constant} and the \textit{energy-momentum tensor}, respectively.
In particular, when $T=0,\ \Lambda=0$, the equation (\ref{Ein_eq}) is called the \textit{vacuum} Einstein equation, then $(M,g_{M})$ is called a \textit{vacuum solution}.
The spacetime $(M,g_{M})$ is a vacuum solution if and only if it is Ricci-flat.

Let $I \subset \mathbb{R}$ be an open interval. 
A spacetime $(M,g_{M})$ is \textit{globally hyperbolic} if there exist an $n$-dimensional manifold $\Sigma$ and a positive function $N$ on $M$ such that $(M,g_{M})$ is isometric to
\begin{equation}
\left(I \times \Sigma, -N^2 dt^2 + h_{t}\right), \label{glob_hyp_s_t}
\end{equation}
where $t$ denotes the coordinate of $I$, and $h_{t}$ denotes a Riemannian metric on $\Sigma$ for all $t \in I$.
Moreover, we call $\Sigma_{t} := \{t\} \times \Sigma$ a \textit{Cauchy hypersurface} for each $t \in I$. 

A globally hyperbolic spacetime (\ref{glob_hyp_s_t})
is \textit{spatially homogeneous} if for each $t_{0} \in I$ and any two points $x, y \in \Sigma_{t_{0}}$ there exists an isometry $f$ of $(M,g_{M})$ 
such that $f(t,x)=(t,y)$ for all $t \in I$. 
Also, a globally hyperbolic spacetime (\ref{glob_hyp_s_t}) is \textit{spatially isotropic} 
if for each $(t_{0},x) \in I \times \Sigma$ 
and any two spacelike vectors $v, w \in T_{x}\Sigma_{t_{0}}$ 
satisfying $g_{M}(v,v) = g_{M}(w, w)$ 
there exists an isometry $f$ of $(M,g_{M})$ such that 
$f(t,x)=(t,x)$ for all $t \in I$ and $f_{\ast}(v) = w$.
By definition, when a globally hyperbolic spacetime is spatially homogeneous, or spatially isotropic, then the function $N$ is a positive function on $I$, which is called the \textit{lapse function} of the spacetime. 

Unless otherwise noted, 
we deal with the case that the Cauchy hypersurface is a connected Lie group $G$, 
that is, 
\begin{equation*}
    (M,g_{M}) = \left(I \times G, -N^2dt^2 + g_{G}(t)\right), 
\end{equation*}
where $N$ is a positive function on $I$ and $g_{G}(t)$ is left-invariant Riemannian metric on $G$ for all $t \in I$.
By definition, the globally hyperbolic spacetime is 
spatially homogeneous but not necessarily spatially isotropic.

Let $G$ be a connected Lie group, $\mathfrak{g}$ its Lie algebra of $G$, 
$\{X_{1},\ldots,X_{n}\}$ a set of generators of $\mathfrak{g}$.
Then we define the \textit{structure constants} $C_{ij}^{k}$ of $\mathfrak{g}$ as
\begin{equation*}
[X_{i}, X_{j}] = \sum_{k=1}^{n} C_{ij}^{k} X_{k}.
\end{equation*}
In addition, we define the left-invariant $1$-forms $\{\omega^{1}, \ldots, \omega^{n}\}$ on $G$ as $\omega^{i}(X_{j}) = \delta^{i}_{j}$.
Then a left-invariant Riemannian metric $g_{G}$ on $G$ can be expressed as 
\begin{equation*}
    g_{G}=\sum_{1 \leq i,j \leq n} s_{ij} \, \omega^{i} \otimes \omega^{j}
\end{equation*}
by using a positive-definite symmetric matrix $\left[s_{ij}\right]_{1 \leq i,j \leq n}$. 
Therefore, considering a one-parameter family of metrics $g_{G}(t)$ on the spatially homogeneous spacetime is equivalent to choosing a smooth curve 
\begin{equation}
    C : I \ni t \mapsto C(t) = \left[ s_{ij}(t) \right] \in \textrm{Sym}_{n}^{+}\mathbb{R}, \label{met_curv}
\end{equation}
where $\textrm{Sym}_{n}^{+}\mathbb{R}$ denotes the set of $n \times n$ positive-definite symmetric matrices.

In Section~\ref{sdiag}, we show that the metrics for the Lie groups considered in this paper can be diagonalized. 
Let $N, a_{1}, \ldots, a_{n}$ be positive functions on an open interval $I \subset \mathbb{R}$.
Consider a spatially homogeneous spacetime with the diagonal Lorentzian metric 
\begin{equation*}
    (M,g_{M}) = \left(I \times G, -N^2dt^2+a_{1}^{2}(\omega^{1})^{2} + \cdots + a_{n}^{2}(\omega^{n})^{2}\right), 
\end{equation*}
where $a_{1}, \ldots, a_{n}$ are called \textit{scale factors}. 
When we give global vector fields on $M$ as 
\begin{equation*}
X_{0} = \frac{d}{dt}, \ X_{1}, X_{2}, \ldots, X_{n},
\end{equation*}
and set 
\begin{equation}
e_{0} = \frac{1}{N}X_{0}, \ e_{1} = \frac{1}{a_{1}}X_{1}, \ 
e_{2} = \frac{1}{a_{2}}X_{2}, \ldots, e_{n} = \frac{1}{a_{n}}X_{n}, \label{ONF}
\end{equation}
then $\{e_{0}, e_{1}, e_{2}, \ldots, e_{n} \}$ are global orthonormal frame fields of $(M, g_{M})$, that is, it holds that 
\begin{equation*}
g_{M}(e_{0}, e_{0}) = -1, \quad g_{M}(e_{i}, e_{i}) = 1 \quad (1 \leq i \leq n), \quad 
\textrm{otherwise} = 0. 
\end{equation*}
Let $\nabla$ be the Levi-Civita connection of $(M, g_{M})$.
By direct computation using the Koszul formula
\begin{align*}
2g_{M}(\nabla_{X}Y, Z) 
&= X(g_{M}(Y, Z)) + Y(g_{M}(Z, X)) - Z(g_{M}(X, Y)) \\
& \qquad - g_{M}(X, [Y, Z]) + g_{M}(Y, [Z, X]) + g_{M}(Z, [X, Y]),
\end{align*}
we have the following.

\begin{prop} \label{LC_conn} 
Let $1 \leq i, j \leq n$. Then 
\begin{align*}
\nabla_{X_{0}}X_{0} &= X_{0}(\log{N})X_{0}, \quad 
\nabla_{X_{0}} X_{i} = \nabla_{X_{i}} X_{0} = X_{0}(\log{a_{i}})X_{i}, \\
\nabla_{X_{i}} X_{j} &= \frac{1}{2} \frac{X_{0}(a_{i}^{2})}{N^{2}} \delta_{ij} X_{0} 
+ \frac{1}{2} \sum_{k=1}^{n} \frac{C_{kj}^{i}a_{i}^{2} + C_{ki}^{j}a_{j}^{2} + C_{ij}^{k}a_{k}^{2}}{a_{k}^{2}}X_{k}.
\end{align*}
\end{prop}

Here we define the $(0,4)$-type Riemannian curvature tensor $R$ as 
\begin{equation*}
R(X,Y,Z,W) = g_{M}(R(X,Y)Z,W) 
= g_{M}(([\nabla_{X}, \nabla_{Y}] - \nabla_{[X, Y]})Z, W).
\end{equation*}
Then we obtain the following Lemma~\ref{Riem_curv} through calculations using Proposition~\ref{LC_conn} and the symmetry of the structure constants 
\begin{equation*}
C_{ij}^{k} = -C_{ji}^{k}, \quad 
\sum_{m=1}^{n} (C_{mi}^{l}C_{jk}^{m} + C_{mj}^{l}C_{ki}^{m} + C_{mk}^{l}C_{ij}^{m}) = 0.
\end{equation*}

\begin{lem} \label{Riem_curv} 
Let $1 \leq i, j, k \leq n$. Then 
\begin{align*}
R(X_{0},X_{i},X_{j},X_{0}) =& \left[X_{0}(\log{N})X_{0}(\log{a_{i}}) - X_{0}^{2}(\log{a_{i}}) - \{X_{0}(\log{a_{i}}) \}^{2} \right]a_{i}^{2} \delta_{ij}, \\
R(X_{k},X_{i},X_{0},X_{k}) =& 
\left[X_{0}(\log{a_{k}}) - X_{0}(\log{a_{i}}) \right] C_{ik}^{k}a_{k}^{2}, \\
R(X_{k},X_{i},X_{j},X_{k}) =& \frac{1}{2N^{2}}
\left[X_{0}(a_{i}^{2})X_{0}(\log{a_{k}})a_{k}^{2}\delta_{ij} - 
X_{0}(a_{k}^{2})X_{0}(\log{a_{i}})a_{i}^{2}\delta_{ik}\delta_{kj} \right] \\
& + \frac{1}{4}\sum_{l=1}^{n}\left[\frac{2C_{kl}^{k}a_{k}^{2}(C_{lj}^{i}a_{i}^{2}+C_{li}^{j}a_{j}^{2})}{a_{l}^{2}} + \frac{(C_{kl}^{i}a_{i}^{2}+C_{il}^{k}a_{k}^{2})(C_{kl}^{j}a_{j}^{2}+C_{jl}^{k}a_{k}^{2})}{a_{l}^{2}} \right.\\ 
& \left.  \qquad \qquad 
+ \, C_{kl}^{i}C_{jk}^{l}a_{i}^{2} + C_{kl}^{j}C_{ik}^{l}a_{j}^{2} 
+ (C_{il}^{k}C_{kj}^{l} + C_{jl}^{k}C_{ki}^{l})a_{k}^{2} 
- 3C_{ki}^{l}C_{kj}^{l}a_{l}^{2} \right],
\end{align*}
where for a function $f$ on $I$
\begin{equation*}
X_{0}^{2}f := X_{0}(X_{0}f) = \frac{d^{2}\!f}{dt^{2}}. 
\end{equation*}
\end{lem}

Nextly we define the Ricci tensor $\mathrm{Ric}$ with respect to $R$ 
of the spatially homogeneous spacetime $(M, g_{M})$ as
\begin{equation*}
\textrm{Ric}(X,Y) = - R(e_{0},X,Y,e_{0}) + \sum_{i=1}^{n} R(e_{i},X,Y,e_{i}), 
\end{equation*}
where $\{e_{0}, e_{1}, \ldots, e_{n}\}$ denotes an orthonormal frame field and $X$ and $Y$ are vector fields on $M$. 
Then we can determine all components of the Ricci tensor by using the orthonormal frame fields (\ref{ONF}) and Lemma~\ref{Riem_curv}.

\begin{prop} \label{Ricci_curv}  
Let $1 \leq i, j \leq n$. Then 
\begin{align*}
\mathrm{Ric}(X_{0},X_{0}) 
=& \sum_{k=1}^{n} \left(\frac{\dot{N}}{N}\frac{\dot{a}_{k}}{a_{k}} -  \frac{\ddot{a}_{k}}{a_{k}}\right), \quad 
\mathrm{Ric}(X_{0},X_{i}) = \frac{d}{dt} 
\log{\left[ \, \prod_{k=1}^{n} \left(\frac{a_{k}}{a_{i}} \right)^{C_{ik}^{k}}\right]}, \\ 
\mathrm{Ric}(X_{i},X_{j}) =& \frac{a_{i}^{2}}{N^{2}}\left[ \left(\sum_{\substack{k=1\\k \neq i}}^{n} \frac{\dot{a}_{k}}{a_{k}} - \frac{\dot{N}}{N} \right)\frac{\dot{a}_{i}}{a_{i}} + \frac{\ddot{a}_{i}}{a_{i}} \right]\delta_{ij} \\
&+ \frac{1}{4} \sum_{k,l=1}^{n}\left[ 2C_{il}^{k}C_{kj}^{l}+ 2\frac{C_{lk}^{k}(C_{jl}^{i}a_{i}^{2} + C_{il}^{j}a_{j}^{2})}{a_{l}^{2}} \right.\\
&+\left.  \frac{(C_{kl}^{i}a_{i}^{2}+C_{il}^{k}a_{k}^{2})(C_{kl}^{j}a_{j}^{2}+C_{jl}^{k}a_{k}^{2})}{a_{k}^{2}a_{l}^{2}} 
+ 
\frac{C_{kl}^{i}C_{jk}^{l}a_{i}^{2} + C_{kl}^{j}C_{ik}^{l}a_{j}^{2} -3C_{ki}^{l}C_{kj}^{l}a_{l}^{2}}{a_{k}^{2}}
\right],
\end{align*}
where for a function $f$ on $I$
\begin{equation*}
\dot{f} := X_{0}f = \frac{df}{dt}, \quad 
\ddot{f} := X_{0}^{2}f = \frac{d^{2}\!f}{dt^{2}}.
\end{equation*}
\end{prop}

By the assumption of spatial homogeneity, the Einstein equations become ODEs.
However, these equations are still difficult to be solved. Even in the four-dimensional case, the general solutions are found only for three cases (see Table~\ref{tBianchi}).
All of these cases correspond to almost abelian Lie groups, that we describe below.

\subsection{Almost abelian Lie groups}
Let $\mathfrak{g}$ be an $n$-dimensional real Lie algebra. 
It is \textit{almost abelian} if there exists an $(n-1)$-dimensional abelian ideal $\mathfrak{a} \subset \mathfrak{g}$ \cite{MR4533011}.
A connected $n$-dimensional Lie group $G$ is \textit{almost abelian} if the Lie algebra $\mathfrak{g}$ of $G$ is almost abelian.

By definition, when $\mathfrak{g}$ is an $n$-dimensional almost abelian Lie algebra, there exists a set of generators $\{X_{1}, \ldots, X_{n-1}, X_{n}\}$ and an 
$(n-1) \times (n-1)$ square matrix $[A_{ij}]_{1 \leq i,j \leq n-1}$ such that 
\begin{equation}
[X_{n}, X_{j}] = \sum_{k=1}^{n-1} A_{kj} X_{k}, 
\quad \textrm{otherwise} = 0. \label{Lie_alg_str}
\end{equation}
Conversely, for a real vector space $V$ with a basis $\{X_{1}, \ldots, X_{n-1}, X_{n}\}$ and an $(n-1) \times (n-1)$ square matrix $[A_{ij}]_{1 \leq i,j \leq n-1}$, when we define a Lie algebra structure on $V$ by the above relations (\ref{Lie_alg_str}), then $\mathfrak{g}_{A} := (V, [\, , \,])$ is almost abelian, 
and we call $A$ the associated matrix for $\mathfrak{g}_{A}$. 
For matrices $A, B$, the almost abelian Lie algebra $\mathfrak{g}_{A}$ is isomorphic to $\mathfrak{g}_{B}$ if and only if $A$ is similar to $B$ up to scaling; that is, 
there exist an invertible matrix $P$ and a non-zero real constant $c$ such that
$B = cP^{-1}AP$. 
See Table~\ref{tBianchi} for three-dimensional examples of $A$.

We consider the case where a Lie group $G$ of the spatially homogeneous spacetime $(M,g_{M})$ is almost abelian. In other words, we suppose that the structure constants are zero except for 
\begin{align} 
C^{i}_{nj} = -C^{i}_{jn} = A_{ij},\ (i,j=1,\dots,n-1).
\end{align}
By applying the structure constants to Proposition~\ref{Ricci_curv}, the Ricci tensor is simplified as follows.

\begin{cor} \label{alm_abel_Ric}  
When $G$ is almost abelian, 
the Ricci tensor of the spatially homogeneous spacetime $(M, g_{M})$ is as follows: 
\begin{align*}
\mathrm{Ric}(X_{0},X_{0}) 
=& \sum_{k=1}^{n} \left(\frac{\dot{N}}{N}\frac{\dot{a}_{k}}{a_{k}} 
-  \frac{\ddot{a}_{k}}{a_{k}}\right), \\ 
\mathrm{Ric}(X_{0},X_{n}) =& \frac{d}{dt} 
\log{\left[ \, \prod_{k=1}^{n-1} \left(\frac{a_{k}}{a_{n}} \right)^{A_{kk}}\right]}, \\ 
\mathrm{Ric}(X_{n},X_{n}) =& \frac{a_{n}^{2}}{N^{2}} 
\left[ \left(\sum_{k=1}^{n-1} \frac{\dot{a}_{k}}{a_{k}} 
- \frac{\dot{N}}{N} \right)\frac{\dot{a}_{n}}{a_{n}} 
+ \frac{\ddot{a}_{n}}{a_{n}} \right] 
- \frac{1}{2}\sum_{k,l=1}^{n-1}\left\{A_{lk}A_{kl} + \left(A_{kl} \right)^{2}\frac{a_{k}^{2}}{a_{l}^{2}} \right\}, \\
\mathrm{Ric}(X_{i},X_{j}) =& \frac{a_{i}^{2}}{N^{2}}\left[ \left(\sum_{\substack{k=1\\k \neq i}}^{n} \frac{\dot{a}_{k}}{a_{k}} - \frac{\dot{N}}{N} \right)\frac{\dot{a}_{i}}{a_{i}} 
+ \frac{\ddot{a}_{i}}{a_{i}} \right]\delta_{ij} \\
& \qquad - \frac{\mathrm{tr} A}{2} \frac{A_{ij}a_{i}^{2}+A_{ji}a_{j}^{2}}{a_{n}^{2}} + \frac{1}{2}\sum_{k=1}^{n-1} \frac{A_{ik}A_{jk}a_{i}^{2}a_{j}^{2} - A_{ki}A_{kj}a_{k}^{4}}{a_{n}^{2}a_{k}^{2}},
\end{align*}
where $1 \leq i, j \leq n-1$. 
The remaining components identically vanish. 
\end{cor}

In the $n=3$ case, the exact solutions have been found only for $G=\mathbb{R}^{3}, \ 
\mathbb{R}\mathrm{H}^{3}, \ H_{3}$ cases.
All of these have zero-dimensional moduli spaces \cite{MR1958155, MR2783396}.
This fact is generalized to arbitrary dimensions as below.
\begin{thm}[\cite{MR1958155, MR2783396}] \label{classification}
Let $G$ be an $n$-dimensional simply-connected, connected Lie group and $\mathfrak{g}$ its Lie algebra. 
The moduli space of left-invariant metrics are all zero-dimensional if and only if the Lie group $G$ is isomorphic to one of the following: 
\begin{equation*}
\mathbb{R}^{n} \ (n \geq 1), \quad 
\mathbb{R}\mathrm{H}^{n} \ (n \geq 2), \quad 
H_{3} \times \mathbb{R}^{n-3} \ (n \geq 3),
\end{equation*}
where $\mathbb{R}\mathrm{H}^{n}$ denotes the $n$-dimensional real hyperbolic space with a Lie group structure, and $H_{3}$ denotes the three-dimensional Heisenberg group 
\begin{equation*}
H_{3} = \left \{ 
\begin{bmatrix}
1 & x & z \\
0 & 1 & y \\
0 & 0 & 1 \\
\end{bmatrix} 
\, \middle| \ 
x, y, z \in \mathbb{R} \right \}.
\end{equation*}
\end{thm}

The Lie groups given in Theorem~\ref{classification} are almost abelian, and the associated matrices are, up to scaling, 
similar to $O_{n-1}, E_{n-1}$ and $J_{2}(0) \oplus O_{n-3}$, respectively.
When $n=3$, these correspond to 
Bianchi type~I, V and II, respectively.

The proposition below is useful for obtaining metrics in the coordinate basis.
\begin{prop} \label{plinv}
The left-invariant basis and the left-invariant forms of almost abelian algebras can be explicitly given in the coordinates as \begin{align}
    X_j &= \sum_{k=1}^{n-1}(e^{A x_n})_{kj}\partial_k,\quad X_n=\partial_n,\\
    \omega^j &= \sum_{k=1}^{n-1}(e^{-A x_n})_{jk}dx_k,\quad \omega^n=dx_n, \label{eome}
\end{align}
where $j=1,\dots, n-1$, $\partial_\mu := \frac{\partial}{\partial x_\mu}$.
\end{prop}
\begin{proof}
We can directly check that these $X_j$ and $\omega^j$ satisfies the commutation relation \eqref{Lie_alg_str} and $\omega^k(X_l)=\delta^{k}_{l}$. 
\end{proof}

\section{Higher-dimensional vacuum solutions}
\label{s:dim}

\subsection{$G=\mathbb{R}^n,\ \mathbb{R}\mathrm{H}^n$ cases}
First, we consider an $(n+1)$-dimensional ($n\ge 1$) spatially homogeneous spacetime 
$M = I \times G$ whose Lie group $G$ is 
an abelian group $\mathbb{R}^{n}$.
The Lie group is trivially almost abelian and 
the associated matrix is the zero matrix $A = O_{n-1}$. 
It is well-known that except for Minkowski spaces 
the exact vacuum solution is given by 
\begin{equation}
    (I \times \mathbb{R}^{n}, -dt^{2} + t^{2p_{1}} dx_{1}^{2} + \cdots + t^{2p_{n}}dx_{n}^{2} ), \label{eq:kas}
\end{equation}
where $p_{1} + \cdots + p_{n} = p_{1}^{2} + \cdots + p_{n}^{2} = 1$, 
and this is called the Kasner solution \cite{Chodos:1979vk}. 

Next, we consider that the Lie group $G$ is 
a real hyperbolic space $\mathbb{R}\mathrm{H}^{n}$ ($n\ge 2$). 
The Lie group is almost abelian and the associated matrix 
is similar to the identity matrix $A = E_{n-1}$ up to scaling. 
The $(n+1)$-dimensional general solution of this type was given in \cite[(14)]{Demaret:1985pt}:
\begin{equation}
\begin{split}
&ds^2 = (\sinh{t})^\frac{2}{n-1}\left[-dt^2+dx_n^2+\sum_{i=1}^{n-1}(\tanh{t}/2)^{2\alpha_i}e^{-\frac{x_n}{n-1}}dx_i^2\right],\\
&\sum_{i=1}^{n-1}\alpha_i = 0,\quad \sum_{i=1}^{n-1}\alpha_i^2 =\frac{n}{n-1},\quad n\geq 3,
\end{split}
\label{eq:solV}
\end{equation}
where $\alpha_i$ are real constants.
This type of solution does not exist for $n=2$.
In the $n=3$ case, it matches the solution in \cite{MR191573}.
In the $n=4$ case, it matches \cite[(4.23)]{Christodoulakis:2003wy}.
The other special solution is \cite[(13)]{Demaret:1985pt}:\footnote{It can be obtained by substituting $P_i=1/(n-1)$ in \cite[(13)]{Demaret:1985pt} and changing coordinates appropriately.}
\begin{align}
ds^2 = -dt^2+\frac{t^2}{x_n^2}\sum_{i=1}^{n}dx_i^2. \label{eq:mil}
\end{align}
It is the Milne solution generalized to $n\geq 2$.

The solutions \eqref{eq:solV} and \eqref{eq:mil} are obtained in \cite{Demaret:1985pt} by assuming that the metrics are diagonal.
We show that we can always diagonalize the metrics without loss of generality in Section~\ref{sdiag}.

\subsection{$G=H_{3} \times \mathbb{R}^{n-3} $ case} \label{ss:aj02}
We consider an $(n+1)$-dimensional spatially homogeneous spacetime $M = I \times G$ whose Lie group $G$ is the product $H_{3} \times \mathbb{R}^{n-3}$ of a three-dimensional Heisenberg group and an abelian group. 
The Lie group is almost abelian and the associated matrix is similar to  $J_{2}(0) \oplus O_{n-3}$ with the one off-diagonal component up to scaling. 
More explicitly, we consider the $(n-1)\times (n-1)$ matrix 
\begin{align} 
A = J_{2}(0) \oplus O_{n-3} = 
\begin{bmatrix}
0 & 1 &&&\\
0&0&&&\\
&&0 &&\\
&&&\ddots &\\
&&&& 0
\end{bmatrix}.
\end{align}

We solve the Einstein equations below. 
We find that the Ricci tensor is simplified by choosing
\begin{align} 
N=c_0\prod_{i=1}^n \frac{a_i}{c_i},
\end{align}
where $c_i:=a_i(0)$ are positive constants.
We can choose $c_i=c\ (i=0,\dots, n)$ by the isometry and $c=1$ by rescaling coordinates, but we leave it to compare with the solutions that have already been found in four- and five-dimensional cases.
We substitute this $N$ and $A$ into the Ricci tensor given in Corollary~\ref{alm_abel_Ric}.
Since $A$ has no diagonal components, $\text{Ric}(X_0,X_i)\ (i=1,\dots,n-1)$  are trivially zero. From the other components, the vacuum Einstein equations become
\begin{align}
\text{tr}\dot{H} &= (\text{tr}H)^2-\text{tr}H^2,\label{eq:r0} \\
\dot{H_1}&=-\frac{a_1^2N^2}{2a_2^2a_n^2}, \label{eq:r11}\\
\dot{H_2}&=\frac{a_1^2N^2}{2a_2^2a_n^2}, \label{eq:r22}\\
\dot{H_i}&=0 \quad (i=3,\dots, n-1),  \label{eq:rii}\\
\dot{H_n}&=\frac{a_1^2N^2}{2a_2^2a_n^2} \label{eq:rnn}, 
\end{align}
where we defined $H_{i} := \dot{a}_{i}/a_{i}$ and $H := \mathrm{diag}(H_{1}, \ldots, H_{n})$. 
The equation \eqref{eq:rii} can be directly integrated. 
By summing \eqref{eq:r11} and \eqref{eq:r22}, or \eqref{eq:r11} and \eqref{eq:rnn}, these can be also integrated. 
The result is, by writing the initial condition as $h_i:=H_i(0)$,
\begin{equation}
\begin{split}
H_1+H_2 &= h_1+h_2,\\
H_1+H_n &= h_1+h_n,\\
H_i &= h_i.
\end{split}
\end{equation}
By integrating these equations again,
\begin{equation}
\begin{split}
a_1 a_2 &= c_1 c_2e^{(h_1+h_2)t},\\
a_1 a_n &= c_1 c_ne^{(h_1+h_n)t},\\
a_i &= c_ie^{h_i t}.    
\label{eq:a1a2c}
\end{split}
\end{equation}
By substituting these equations into (\ref{eq:r11}), it become a differential equation of $a_1$ only:
\begin{equation}
\begin{split}
&\dot{H_1}+\frac{c'}{2}a_1^4 e^{2h't}=0\quad \left(h':= \sum_{i=3}^{n-1}h_i,\ c':= \frac{c_0^2}{c_1^2 c_2^2 c_n^2}\right).  
\end{split}
\end{equation}
The solution of this differential equation is:
\begin{align} 
a_1 &= \frac{k^{1/2}}{(c')^{1/4}} e^{-\frac{h'}{2} t}\cosh^{-\frac{1}{2}}[k(t-t_0)],
\end{align}
where $k,\ t_0$ are integration constants. By shifting time, we choose $t_0=0$, then
\begin{align} 
a_1 &= c_1 e^{h_1 t}\cosh^{-\frac{1}{2}}kt,\label{eq:a1c}\\
k &= \frac{c_0 c_1}{c_2 c_n}, \quad
h' = -2h_1. \label{eq:kc0}
\end{align}
By substituting it to (\ref{eq:a1a2c}), the other scale factors are
\begin{align}
    a_2&=c_2e^{h_2 t}\cosh^{\hf}kt,\\
    a_n&=c_ne^{h_n t}\cosh^{\hf}kt.
\end{align}

Finally, we substitute them into the Einstein equation (\ref{eq:r0}). 
The result is a constraint on the initial conditions:
\begin{align} 
(\tr h)^2-\tr(h^2)=\frac{k^2}{2}, \label{eq:h2h}
\end{align}
where $h := \text{diag}(h_1,\dots h_n)$.

From the general formula \eqref{eome}, the left-invariant forms are:
\begin{align} 
\omega^1=dx_1 -x_n dx_2, \quad \omega^i = dx_i \ (i=2,3,\dots n).
\end{align}
By using these forms, scale factors $a_i$ derived above, and collecting the conditions (\ref{eq:kc0}), (\ref{eq:h2h}), we obtain the solution (\ref{eq:ds2}). 

Our solution matches with previous studies on the $n=3,4$ cases. 
In the $n=3$ case, by choosing
\begin{align} 
c_0=\frac{1}{c_1}=c_2=c_3=\frac{1}{\sqrt{k}},
\end{align}
our solution realizes the Bianchi type II solution \cite{Taub:1950ez} (in the notation of \cite[(13.55)]{MR2003646}).
In the $n=4$ case, our result includes \cite[(4.5)]{Christodoulakis:2003wy} as a case of
\begin{align} 
c_0=\frac{1}{c_1},\ c_2=c_4=\frac{1}{\sqrt{k}}, \ c_3=1.
\end{align}

\begin{prop} \label{sp_anistrp}
The vacuum solution {\rm (\ref{eq:ds2})} is spatially homogeneous but not spatially isotropic. 
\end{prop}

\begin{proof}
For a diagonal spatially homogeneous spacetime $M = I \times G$, 
since $e_{0}$ is the unit timelike vector of the Cauchy hypersurface 
$G_{t} = \{t\} \times G \subset M$ and from Proposition~\ref{LC_conn}, we have 
\begin{equation*}
    \nabla_{e_{i}}e_{0} = \frac{1}{N(t)a_{i}(t)} \nabla_{X_{0}}X_{i} 
    = \frac{H_{i}(t)}{N(t)} e_{i}.
\end{equation*}
If the diagonal spatially homogeneous spacetime is spatially isotropic, 
then by definition, we have for all $t \in I$ 
\begin{equation}
H_{1}(t) = \cdots = H_{n}(t) \label{tot_umb}
\end{equation}
for the orthonormal frame $\{e_{1}, \ldots, e_{n} \}$ of $G_{t}$.

For $A = J_{2}(0) \oplus O_{n-3}$, 
the equations~(\ref{tot_umb}) do not hold. 
\end{proof}

\section{Diagonalizability of metrics} \label{sdiag}
The $n$-dimensional solutions for the cases $G=\mathbb{R}^n$ \cite{Chodos:1979vk}, $\mathbb{R}\mathrm{H}^n$ \cite{Demaret:1985pt} and $H_3\times \mathbb{R}^{n-3}$ \eqref{eq:ds2} are all diagonal metrics.
In this section, we prove that it is enough to consider diagonal metrics for these Lie groups.

Let $G$ be a connected Lie group, $\mathfrak{g}$ its Lie algebra and $\{X_{1}, \ldots, X_{n}\}$ a set of generators of $\mathfrak{g}$.
Then with respect to the generators, we define the automorphism group of $\mathfrak{g}$ as
\begin{equation*}
    \textrm{Aut}(\mathfrak{g}) = 
    \left\{ \Phi=\left[\varphi_{j}^{i}\right] \in 
    \textrm{GL}_{n}\mathbb{R} \ \middle| \ 
    \sum_{l=1}^{n} C_{ij}^{l} \varphi_{l}^{k} = \sum_{l,h=1}^{n} \varphi_{i}^{l}\varphi_{j}^{h}C_{lh}^{k} \ (1 \leq i,j,k \leq n) \right \},
\end{equation*}
where $C_{ij}^{k}$ is the structure constants determined from the generators.

For a general spatially homogeneous spacetime
\begin{equation*}
    (M, g_{M}) = \left(I \times G, -N^2 dt^2 + \sum_{i,j=1}^{n} s_{ij}(t) \,  \omega^{i} \otimes \omega^{j} \right),
\end{equation*}
we consider the following transformation
\begin{equation}
[\overline{s}_{ij}(t)] = \Phi^{T} [s_{ij}(t)] \Phi \quad 
(\Phi \in \textrm{Aut}(\mathfrak{g})). \label{trans}
\end{equation}
Then the above spatially homogeneous spacetime $(M, g_{M})$ is isometric to
\begin{equation*}
    \left(I \times G, -N^{2} dt^{2} + \sum_{i,j=1}^{n} \overline{s}_{ij}(t) \, \omega^{i} \otimes \omega^{j} \right)
\end{equation*}
as Lorentzian manifolds since $\Phi$ preserves the structure constants.

Let $\{X_{1}, \ldots, X_{n}\}$ be a certain set of generators of a Lie algebra $\mathfrak{g}$. 
We consider the following conditions: 
For any positive-definite symmetric matrix $S \in \textrm{Sym}_{n}^{+}\mathbb{R}$, 
there exists an automorphism $\Phi \in \textrm{Aut}(\mathfrak{g})$ such that 
\begin{equation}
    \Phi^{T} S \Phi = \textrm{diag}(\lambda_{1}^{2}, \ldots, \lambda_{n}^{2}), \label{diag_cond}
\end{equation}
where $\lambda_{1}, \ldots, \lambda_{n}$ denote positive numbers. 

\begin{thm} \label{diag}
Let $\mathfrak{g}$ be a Lie algebra. 
Let a set of generators $\{X_{1}, \ldots, X_{n}\}$ satisfy 
the above condition (\ref{diag_cond}). 
Then for a general spatially homogeneous spacetime 
\begin{equation*}
    (M, g_{M}) = \left(I \times G, -N(t)^{2} dt^{2} + \sum_{1 \leq i,j \leq n} s_{ij}(t) \, \omega^{i} \otimes \omega^{j} \right), 
\end{equation*} 
if it is Ricci-flat, $(M,g_{M})$ is locally isometric to 
a diagonal spatially homogeneous spacetime
\begin{equation*}
    \left(I \times G, -N(t)^{2} dt^{2} + a_{1}(t)^{2} (\omega^{1})^{2} + \cdots + a_{n}(t)^{2} (\omega^{n})^{2} \right), 
\end{equation*}
where $\{\omega^{1}, \ldots, \omega^{n} \}$ denote 
left-invariant 1-forms on $G$ 
with respect to $\{X_{1}, \ldots, X_{n}\}$.
\end{thm}

\begin{proof}
First, we fix an initial time $t_{0} \in I$. 
When we use canonical coordinates of the first kind $\{x_{1}, \ldots, x_{n}\}$, the Lorentzian metric is 
\begin{equation*}
(g_{M})_{(t_{0},e)} = -{N(t_{0})}^2 dt^{2} + \sum_{1 \leq i,j \leq n} 
s_{ij}(t_{0}) \, (dx_{i})_{e} \otimes (dx_{j})_{e} 
\end{equation*}
at the point $(t_{0}, e) \in I \times G$. 
Since the generators of $\mathfrak{g}$ satisfies the above condition (\ref{diag_cond}), 
we have 
\begin{equation*}
    s_{ij}(t_{0}) 
    = \lambda_{i}^{2} \delta_{ij} \quad (1 \leq i,j \leq n)
\end{equation*}
by using the transformation (\ref{trans}). 
When we consider 
\begin{equation*}
    C(t) = [s_{ij}(t)]_{1 \leq i,j \leq n} \in \textrm{Sym}_{n}^{+}\mathbb{R}
\end{equation*}
by using the curve (\ref{met_curv}), 
and let $D = \textrm{diag}(\lambda_{1}, \ldots, \lambda_{n})$, 
then we have $C(t_{0}) = D^{2}$. 
Since $D^{-1}\dot{C}(t_{0})D^{-1}$ is a symmetric matrix, 
there exists an orthogonal matrix $Q$ such that 
\begin{equation*}
    Q^{T} D^{-1} \dot{C}(t_{0}) D^{-1} Q 
    = \textrm{diag}(\alpha_{1}, \ldots, \alpha_{n}), 
\end{equation*}
where we remark $\dot{C} = \frac{dC}{dt}$. 
So letting $P = [p_{j}^{i}] = D^{-1} Q D$ and 
replacing the coordinates 
\begin{equation*}
(y_{1}, \ldots, y_{n}) = (x_{1}, \ldots, x_{n})P, 
\end{equation*}
then we have by using the coordinates $\{y_{1}, \ldots, y_{n}\}$ 
\begin{equation*}
(\dot{g}_{M})_{(t_{0}, e)} = - 2N(t_{0})\dot{N}(t_{0}) dt^{2} + 
\sum_{i,j=1}^{n} \lambda^{2}_{i} \alpha_{i} \delta_{ij} (dy_{i})_{e} \otimes (dy_{j})_{e}, 
\end{equation*}
that is, 
\begin{equation*}
g_{G}(t_{0})_{e} = 
\sum_{i=1}^{n} \lambda^{2}_{i}(dy_{i})^{2}_{e} , \quad 
\dot{g}_{G}(t_{0})_{e} = 
\sum_{i=1}^{n} \lambda^{2}_{i}\alpha_{i}(dy_{i})^{2}_{e}. 
\end{equation*}
Since we suppose the spatially homogeneous spacetime satisfies the Ricci-flat condition $\textrm{Ric} = 0$, 
by the uniqueness of initial value problems of ODE,
we may let the spatially homogeneous spacetime be locally a diagonal one.
\end{proof}

For a Lie group $G$, if the dimension of the moduli space is zero, 
there exists a set of generators $\{X_{1}, \ldots, X_{n}\}$ such that 
for any positive-definite symmetric matrix $S \in \textrm{Sym}_{n}^{+}\mathbb{R}$, 
there exists an automorphism $\Phi \in \textrm{Aut}(\mathfrak{g})$ such that 
\begin{equation*}
    \Phi^{T} S \Phi = \textrm{diag}(\lambda^{2}, \ldots, \lambda^{2}), 
\end{equation*}
where $\lambda$ denotes a positive number. 
See \cite[Section~3]{MR2783396} in detail. 
Namely, the Lie group whose moduli space is zero-dimensional satisfies the condition (\ref{diag_cond}).

\section{Expansion and contraction of space}
\label{s:exp}
In this section, we discuss the behavior of the scale factors in the solution (\ref{eq:ds2}). 
In particular, we show that the space cannot expand or contract in all dimensions in the limit $t\to \infty$. 
We choose initial conditions as $a_i(0)=c_i=1$ ($i=0,1,\dots,n$) in this section; it can be done by the isometry and rescaling.

From the solution (\ref{eq:ds2}) and the definition $H_i=\dot{a}_i/a_i$, we obtain
\begin{align}
H_1(\infty) &= h_1-\frac{1}{2}, \\
H_2(\infty) &= h_2+\frac{1}{2},\\
H_i(\infty) &= h_i \quad (3 \leq i \leq n-1),\\
H_n(\infty) &= h_n +\frac{1}{2}, 
\end{align}
where $H_{j}(\infty) := \lim_{t \to \infty} H_{j}(t)$ 
($1 \leq j \leq n$). 
The constraints on the initial condition are 
\begin{align}
&2h_1 =-\sum_{i=3}^{n-1}h_i, \label{eq:h2h1}\\
&(\tr h)^2-\tr(h^2) =\frac{1}{2}. \label{eq:trh}
\end{align}

\begin{prop} 
For the solution {\rm (\ref{eq:ds2})}, $H_i(\infty)\ (i=1,\dots, n)$ cannot be all positive or all negative.
\end{prop}

\begin{proof}
We prove it by contradiction. First, we show that all $H_i(\infty)$ ($i=1,\dots,n$) cannot be positive. If they can,  $h_i>0$ for $i=3,\dots n-1$. 
The right-hand side of the constraint (\ref{eq:h2h1}) is negative, then  $H_1(\infty)<h_1<0$. 
Thus, all dimensions cannot expand.

Next, we show that all $H_i(\infty)$ ($i=1,\dots,n$) cannot be negative.
If they can, $h_2<-1/2$,  $h_i<0\ (i=3,\dots n-1)$, $h_n<-1/2$. These are all negative and also by using (\ref{eq:h2h1}), we obtain
\begin{align*} 
\sum_{i=3}^{n-1}h_i^2 &< \left(\sum_{i=3}^{n-1}h_i\right)^2 \\
&= 4h_1^2.
\end{align*}
The left-hand side of (\ref{eq:trh}) is
\begin{align*} 
(\tr h)^2-\tr(h^2) 
&= (-h_1+h_2+h_n)^2 - \sum_{i=1}^{n}h_i^2\\
&> (-h_1+h_2+h_n)^2 - (5h_1^2+h_2^2+h_n^2)\\
&= -4h_1^2+2(h_1(-h_2-h_n)+h_2 h_n)\\
&> -4h_1^2+2h_1+\frac{1}{2}\\
&> \frac{1}{2}.
\end{align*}
We have used \eqref{eq:h2h1} for the first equality and $h_2,h_n<-1/2$ for the second inequality. For the last inequality, we have used the condition for contraction $h_1<1/2$ and $h_1>0$ from (\ref{eq:h2h1}). The above inequality contradicts the constraint (\ref{eq:trh}).
Thus, all dimensions cannot contract.
\end{proof}

For cosmology, solutions that only three dimensions equally expand are interesting. As an example of the $n=9$ case (the dimension of superstring theories), there is such a solution of \eqref{eq:h2h1} and \eqref{eq:trh}:
\begin{align}
\begin{split}
h_1&=0,\ h_2=h_{9}=-1,\\
h_3&=h_4=h_5=\frac{1}{2},\ h_6=h_7=h_8=-\frac{1}{2}.\\
\end{split}   
\label{eq:h9}
\end{align}
We plot the time evolution of $a_i$ in Figure~\ref{fig:a9}.
Such a solution can be found for the other dimensional cases, including $n=10$ (the dimension of the M-theory).
One example is
\begin{align}
\begin{split}
h_1&=0,\ h_2=h_{10}=-\frac{13}{2},\\
h_3&=h_4=h_5=4,\ h_6=h_7=h_8=h_9=-3.
\end{split}   
\label{eq:h11}
\end{align}
We plot the time evolution of $a_i$ in Figure~\ref{fig:a10}.
In these examples, three dimensions ($a_3,a_4,a_5$) equally expand, and the other dimensions contract.

\begin{figure}[htb]
    \centering
    \includegraphics[width=13cm]{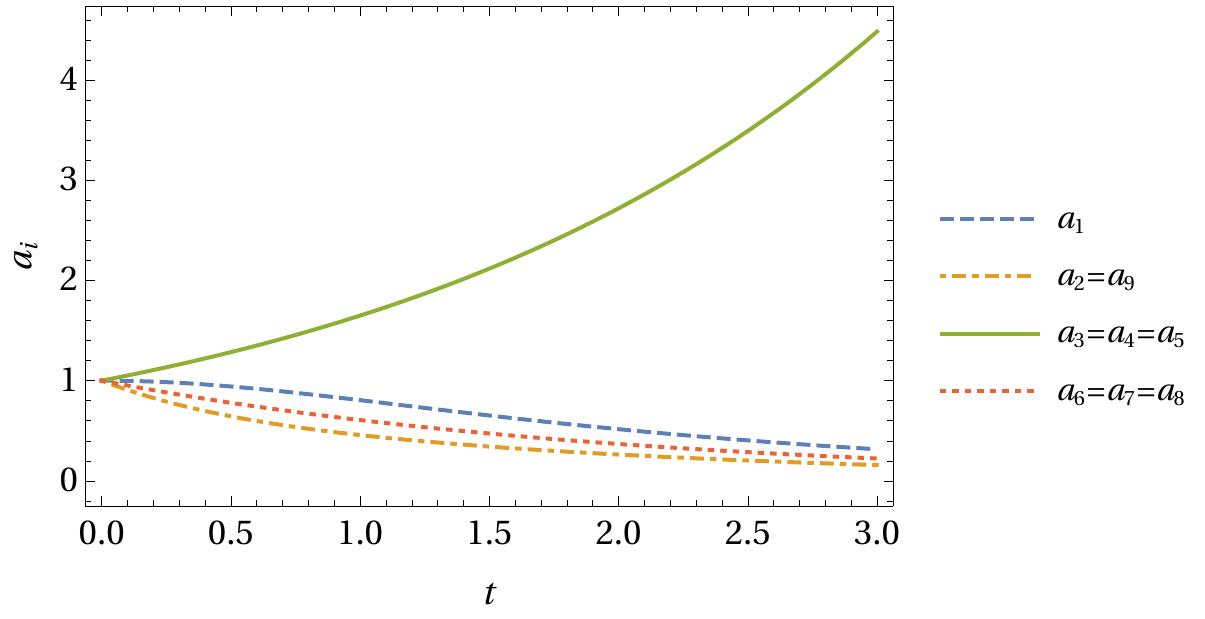}
    \caption{
    Time evolution of the scale factors $a_i(t)$ of the solution (\ref{eq:ds2}) with the spatial dimension number $n=9$. 
    The initial conditions are chosen as $a_i(0) = 1$ 
    and (\ref{eq:h9}).}
    \label{fig:a9}
\end{figure}

\begin{figure}[htb]
    \centering
    \includegraphics[width=13cm]{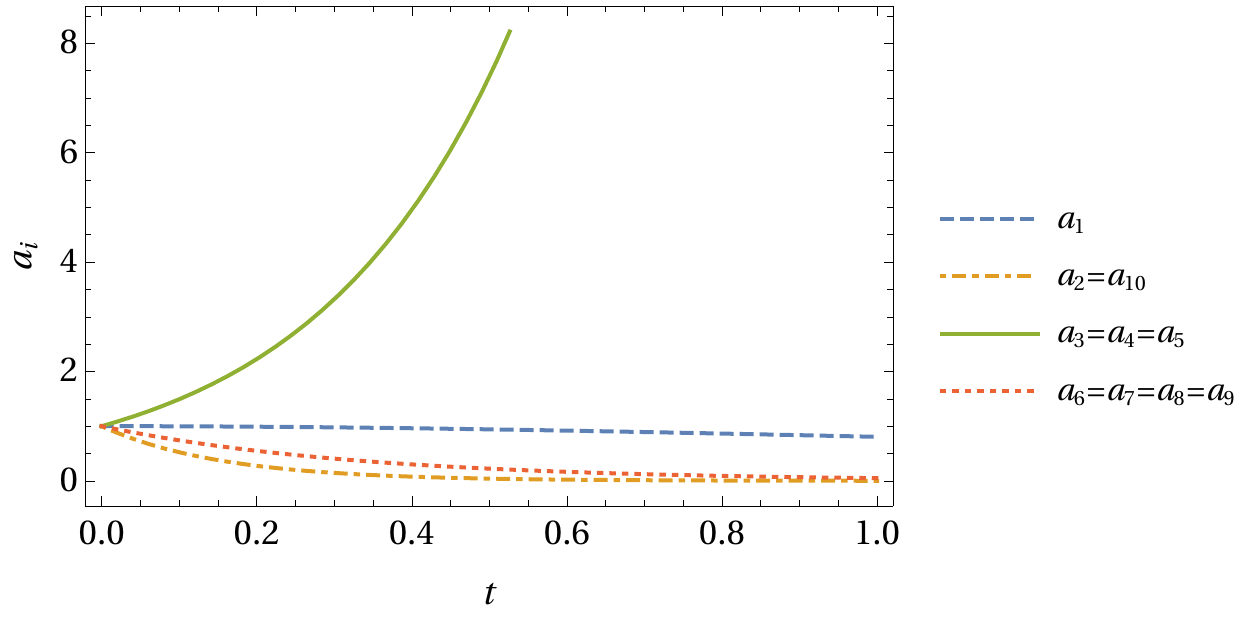}
    \caption{
    Time evolution of scale factors $a_i(t)$ of the solution (\ref{eq:ds2}) with the spatial dimension number $n=10$. 
    The initial conditions are chosen as $a_i(0) = 1$ 
    and (\ref{eq:h11}).}
    \label{fig:a10}
\end{figure}

\section*{Acknowledgments}
The authors would like to thank Tomoya Nakamura for useful discussions and comments. 

\appendix
\section{New solutions for other generalizations of Bianchi type II}
\label{s:oth}

By applying the method of deriving the solution in Subsection~\ref{ss:aj02}, Einstein equations for other spatially homogeneous cases can be solved. As examples, we generalize the Bianchi spacetimes of type II in two different ways below. The first case is that $A$ consists of only off-diagonal components $J_{2}(0)$; in the second case, $A$  consists of  $J_{2}(0)$ and diagonal components. 
When diagonal components are zeros, 
the latter case is reduced to the one off-diagonal component case $J_{2}(0) \oplus O_{n-3}$ in Subsection~\ref{ss:aj02}. 

We assume that metrics are diagonal. We do not claim that the diagonal solutions in this appendix are general.

\subsection{Block-diagonal case} 
\label{subs:1off}

We consider odd $n\geq 3$, and
\begin{align} 
A= \underbrace{J_{2}(0) \oplus \cdots \oplus J_{2}(0)}_{\frac{n-1}{2}} 
= 
\begin{bmatrix}
0 & 1 &&&&&\\
0&0&&&&&\\
&&0&1&&&\\
&&0&0&&&\\
&&&&\ddots&&\\
&&&&&0&1\\
&&&&&0&0
\end{bmatrix}.
\end{align}

The Ricci-flat conditions are: 
\begin{align} 
\dot{H}_{2p-1}&=-\frac{N^2}{2 a_n^2}\frac{a_{2p-1}^2 }{a_{2p}^2}, \label{eq:h2p1}\\
\dot{H}_{2p}&=\frac{N^2}{2 a_n^2}\frac{a_{2p-1}^2}{a_{2p}^2}, \label{eq:h2p}\\
\dot{H}_n&=\frac{N^2}{2 a_n^2}\sum_{k=1}^{\frac{n-1}{2}}\frac{a_{2p-1}^2}{a_{2p}^2}, \label{eq:hnd0}
\end{align}
where $p=1,\dots,(n-1)/2$. 
From the general formula \eqref{eome}, the left-invariant forms are:
\begin{equation}
\begin{split}
\omega^{2p-1}&=dx_{2p-1}-x_n dx_{2p}, \\
\omega^{2p}&=dx_{2p}\quad \left(p=1,\dots \frac{n-1}{2}\right),\\
\omega^n &= dx_n.    
\end{split}
\end{equation}
The solution of the Einstein equations is: 
\begin{equation}
\begin{split}
ds^2=&\cosh^{\frac{n-1}{2}}(kt)(c_0^2 e^{2(\tr h)t}dt^2 +c_n^2 e^{2h_n t}(\omega^n)^2)\\
&+\sum_{p=1}^{\frac{n-1}{2}}\left[\frac{c_{2p-1}^2e^{2h_{2p-1}t}}{\cosh kt}(\omega^{2p-1})^2+c_{2p}^2e^{2h_{2p}t}\cosh(kt)(\omega^{2p})^2\right],\\
k=&\frac{c_0 c_1}{c_2c_n},\ 
\frac{c_1}{c_2}=\frac{c_3}{c_4}=\dots=\frac{c_{n-2}}{c_{n-1}},\ 
(\text{tr}h)^2-\text{tr}(h^2)=\frac{n-1}{4}k^2.
\end{split}
\label{eq:ds2c}
\end{equation}
In the $n=3$ case, it is equivalent to the known solution \cite[(13.55)]{MR2003646}.

\subsection{One off-diagonal and diagonal components}

We consider
\begin{align} 
A= J_{2}(0) \oplus \textrm{diag}(A_{3}, \ldots, A_{n-1}) = 
\begin{bmatrix}
0 & 1 &&&\\
0&0&&&\\
&&A_{3}&&\\
&&&\ddots&\\
&&&&A_{n-1}
\end{bmatrix}.
\end{align}
For such $A$, $\text{Ric}(X_0,X_n)$ and $\text{Ric}(X_1,X_2)$ is nonzero. Then there are two more constraints compared to the previous cases.
The Einstein equations are
\begin{align}
\dot{H_1} &=-\frac{N^2}{2a_n^2}\frac{a_1^{2}}{a_2^2}, \label{eq:h1d}\\
\dot{H_2} &=\frac{N^2}{2a_n^2}\frac{a_1^{2}}{a_2^2}, \label{eq:h2d}\\
\dot{H_i}&=0\quad (i=3,\dots, n-1),  \label{eq:hi}\\
\dot{H_n}&=\frac{N^2}{2a_n^2}\left[2\tr (A^2) +\frac{a_1^{2}}{a_2^2}\right] \label{eq:hn2},\\
0&=\frac{N^2\tr A}{2a_n^2}, \label{eq:0n2}\\
0&=\sum_{i=3}^{n-1}A_{i} H_i - (n-1)H_n \text{tr}A. \label{eq:0i3}
\end{align}
In a similar way as previous subsections, we can obtain $a_i(t)$. By substituting $a_i(t)$ into the condition (\ref{eq:r0}), we obtain a constraint on $h_i$. The additional equations (\ref{eq:0n2}) and (\ref{eq:0i3}) also give constraints on the parameters $h_i$. 
From the general formula \eqref{eome}, the left-invariant forms are:
\begin{align*} 
&\omega^1=dx_1-x_n dx_2,\\
&\omega^2=dx_2,\\
&\omega^i= e^{-A_ix_n}dx_i \quad (i=3,\dots n-1),\\
&\omega^n= dx_n.
\end{align*}
The solution of the Einstein equations is: 
\begin{equation}
\begin{split}
ds^2=&f(t)\cosh(kt)\left[-c_0^2 e^{2(\tr h)t}dt^2 +c_n^2 e^{2h_n t}(\omega^n)^2\right]\\
&+\frac{c_1^2e^{2h_1 t}}{\cosh(kt)}(\omega^1)^2+c_2^2e^{2h_2 t}\cosh(kt)(\omega^2)^2+\sum_{i=3}^{n-1} c_i^2e^{2h_i t}(\omega^i)^2,\\
f(t) :=&  \exp\left\{\frac{\text{tr}(A^2)}{(h_2-h_1)^2}\left[\frac{c_0^2}{2c_n^2}(e^{2(h_2-h_1)t}-1)-(h_2-h_1)t\right]\right\}, \\
k=&\frac{c_0 c_1}{c_2c_n},\  
\text{tr} A= 0,\ 
\sum_{i=3}^{n-1} h_i A_{i} = 0,\\
h_1 =& -\hf\sum_{i=3}^{n-1}h_i,\ 
(\text{tr}h)^2-\text{tr}(h^2)= \frac{c_0^2}{c_n^2}\text{tr}(A^2)+\frac{k^2}{2}.
\end{split}
\label{eq:ds2e}
\end{equation}
The condition $h_1\neq h_2$ is included in the conditions in the last line. It can be shown as follows. If $h_1=h_2$, the left equation of the last line implies $0=\sum_{i=1}^{n-1}h_i$ then $\tr h = h_n$, and $(\text{tr}h)^2-\text{tr}(h^2)<0$. It contradicts the last equation.

As a special case $A_{i}=0\ (i=3,\dots,n-1)$, this solution includes the solution~(\ref{eq:ds2}).

\bibliographystyle{utphys}
\bibliography{ref}

\end{document}